\documentclass[%
reprint,
superscriptaddress,
nofootinbib,
 amsmath,amssymb,
 aps,
pra,
]{revtex4-1}

\usepackage{graphicx}
\usepackage{dcolumn}
\usepackage{bm}
\usepackage{amsmath,amsthm}
\usepackage{amsfonts,amssymb}
\usepackage{mathrsfs}
\usepackage{chngcntr}
\usepackage{enumitem}
\usepackage[dvipsnames]{xcolor}
\usepackage[caption=false]{subfig}
\usepackage{hyperref}
\usepackage{hyphenat}
\counterwithout{equation}{section}

\newcommand{\norm}[1]{\left\lVert#1\right\rVert}
\newcommand{\abs}[1]{\left\lvert#1\right\rvert}
\newcommand{\ccausal}[1]{#1}

\newcommand{\leftcp}{\mathcal{E}_l}
\newcommand{\rightcp}{\mathcal{E}_r}

\def\bra#1{\langle{#1}|}
\def\ket#1{|{#1}\rangle}
\newcommand{\braket}[2]{\left< #1 \vphantom{#2} \right|
 \left. #2 \vphantom{#1} \right>} 

\def\inline#1{\emph{#1}.---}

\newtheorem{thm}{Theorem}
\newtheorem{lem}{Lemma}

\begin{document}

\preprint{APS/123-QED}

\title{Matrix Product States for Quantum Stochastic Modeling}


\author{Chengran Yang}
\email{Yangchengran92@gmail.com}
\affiliation{
School of Physical and Mathematical Sciences, Nanyang Technological University, 637371 Singapore, Singapore
}
\affiliation{Complexity institute, Nanyang Technological University, 639798 Singapore, Singapore}

\author{Felix C. Binder}
\email{quantum@felix-binder.net}
\affiliation{
School of Physical and Mathematical Sciences, Nanyang Technological University, 637371 Singapore, Singapore
}
\affiliation{Complexity institute, Nanyang Technological University, 639798 Singapore, Singapore}

\author{Varun Narasimhachar}
\affiliation{
School of Physical and Mathematical Sciences, Nanyang Technological University, 637371 Singapore, Singapore
}
\affiliation{Complexity institute, Nanyang Technological University, 639798 Singapore, Singapore}

\author{Mile Gu}
\email{mgu@quantumcomplexity.org}
\affiliation{
School of Physical and Mathematical Sciences, Nanyang Technological University, 637371 Singapore, Singapore
}
\affiliation{Complexity institute, Nanyang Technological University, 639798 Singapore, Singapore}
\affiliation{
Centre for Quantum Technologies, National University of Singapore, 3 Science Drive 2, 117543 Singapore, Singapore
}

\date{\today}

\pacs{Valid PACS appear here}

\begin{abstract}
In stochastic modeling, there has been a significant effort towards finding predictive models that predict a stochastic process' future using minimal information from its past. Meanwhile, in condensed matter physics, matrix product states (MPS) are known as a particularly efficient representation of 1D spin chains. In this Letter, we associate each stochastic process with a suitable quantum state of a spin chain. We then show that the optimal predictive model for the process leads directly to an MPS representation of the associated quantum state. Conversely, MPS methods offer a systematic construction of the best known quantum predictive models. This connection allows an improved method for computing the quantum memory needed for generating optimal predictions. We prove that this memory coincides with the entanglement of the associated spin chain across the past-future bipartition.
\end{abstract}
\maketitle

The quest for simple representations and models of the physical world, often phrased as Occam's famous razor, underlies most scientific pursuits. In this spirit, computational mechanics seeks the most memory-efficient predictive models for stochastic processes--models which track relevant past information about a process, in order to generate statistically faithful future predictions~\cite{1986Grassberger,Crutchfield1989,Shalizi2001,Crutchfield2011,Paul2013}. The classically minimal models, $\varepsilon$-machines, have been used in diverse contexts from  neuroscience to nonequilibrium  contextuality~\cite{Haslinger2009,yang2008increasing,park2007complexity,Clarke2003,Varn2002,Varn2013,Varn2015,feldman1998measures,Boyd2016identifying,Garner2017themodynamics}. Recently, it was shown that quantum extensions of $\varepsilon$-machines can further reduce their memory~\cite{gu2012quantum}, leading recent studies to find memory-efficient quantum means of predictive modeling~\cite{riechers2015closed,suen2017classical,aghamohammadi2016ambiguity,aghamohammadi2016extreme,mahoney2016occam,Binder2017,Aghamohammadi2017PRX,Elliott2018Continous,Palssone1601302,Garner2017Unbounded,Jouneghani2017Observing}.

In condensed matter, on the other hand, simplicity is sought after for the description of quantum many-body systems. Tensor networks, such as matrix product states (MPS), for instance, provide an efficient and useful description of one-dimensional quantum systems--i.e., spin chains~\cite{Affleck1987,klumper1993matrix,Prosen2011}. This has led to reliable and powerful numerical methods for probing and simulating properties of multi-partite systems, whose study would be otherwise intractable~\cite{white1992density,verstraete2004density,Vidal2007,Orus2007}.

\begin{figure}
	\includegraphics[scale=0.38]{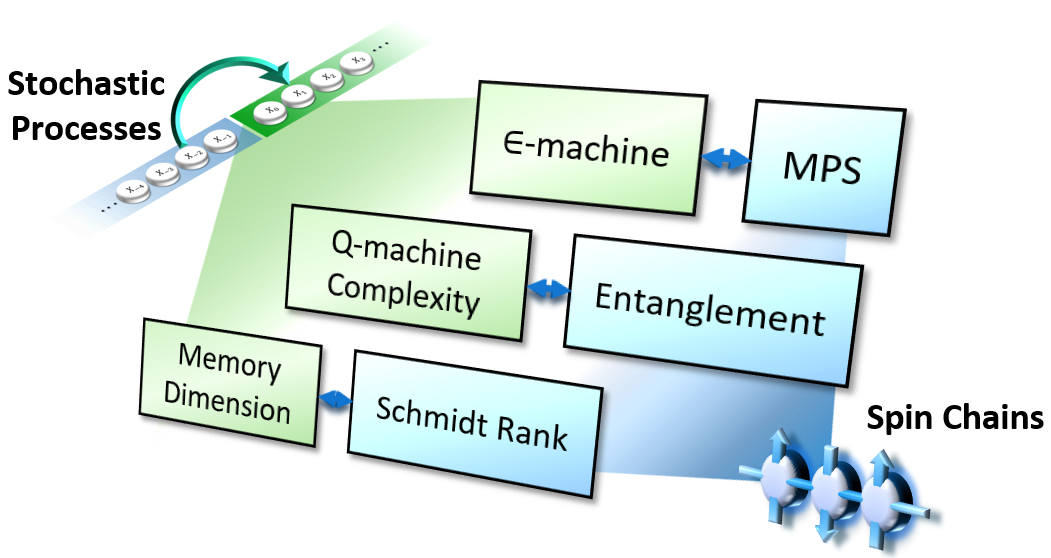}
	\caption{This Letter connects the complexity of stochastic processes and the representational complexity of spin chains. These (i) link $\epsilon$-machines, the provably optimal predictors of a process with the MPS states of a spin chain, (ii) the memory used to simulate a process quantum mechanically with the entanglement of the spin chain and (iii) the Hilbert space dimension of such quantum simulators with the Schmidt rank of the spin chain.}
	\label{fig:diagram}
\end{figure}

In this Letter, we develop a connection between $\varepsilon$-machines and the MPS representation, shown in Fig.~\ref{fig:diagram}. We associate each stochastic process with a suitable quantum state of a spin chain called q-sample, measurement of which generates the corresponding stochastic process. We then show that the classical $\varepsilon$-machine of the process leads to a systematic MPS representation of the q-sample. Conversely, applying MPS methods to this state allows us to construct a q-simulator for the associated stochastic process -- the best known quantum model~\cite{Binder2017,mahoney2016occam,riechers2015closed}. Lastly, we show that the entanglement across the past-future bipartition of the q-sample coincides exactly with the quantum memory requirements of the q-simulator. These results thus provide a direct means of using MPS methods to study the resource requirements of quantum stochastic simulation, extending their relevance to the field of predictive modeling.

Our approach complements other uses of tensor network methods for the description of classical systems with stochastic elements~\cite{Johnson2010,Johnson2015,2015Guta,Monras2016,Critch2012,2014Kliesch,Monras} and machine learning~\cite{Han2017,Novikov2016,Stoudenmire2017,Stoudenmire2016,Loeliger2017,Glasser2018,Guo2018,Chen2018,Huggins2018}. We extend these results in introducing causal structure, adapting MPS methods for predictive modeling.

\inline{Predictive models}Consider a system that generates an output $x_t$ sampled from a random variable $X_t$ at each time $t$. The outputs of the system are described by a stochastic process with joint probability $P(\overleftarrow{X},\overrightarrow{X})$, which correlates past observations $\overleftarrow{X}:=\cdots X_{-2}X_{-1}$ with future outputs $\overrightarrow{X}:= X_{0}X_{1}\cdots$. Thus, each instance of the process with a given past $\overleftarrow{x}:=\cdots x_{-2}x_{-1}$, will exhibit a future output  $\overrightarrow{x}:=x_{0}x_{1}\cdots$ with probability $P(\overrightarrow{x}|\overleftarrow{x})$.

To make predictions, one wants to construct a predictive model that replicates this conditional behavior. Such a model specifies how each observed past $\overleftarrow{x}$ can be encoded into some memory state $\varepsilon(\overleftarrow{x})$ of a physical system $\Xi$, such that repeated systematic actions on $\Xi$ generate outputs $\overrightarrow{x}$ with desired conditional probability, i.e., $P[\overrightarrow{x}|\varepsilon(\overleftarrow{x})] =P(\overrightarrow{x}|\overleftarrow{x})$. Such predictive models are causal--all information about the future contained in $\Xi$ can be obtained from the past.

Storing $\varepsilon(\overleftarrow{x})$, however, costs resources. There is thus interest in encoding strategies that minimize the entropy of $\Xi$. In the case of ergodic stationary processes, computational mechanics provides the provably optimal strategy~\cite{Shalizi2001}. This involves grouping pasts $\overleftarrow{x}$ with identical future statistics $\overrightarrow{X}$ into equivalence classes called \textit{causal states} (labeled $s_i$) by the equivalence relation
\begin{equation}\overleftarrow{x} \sim_\varepsilon \overleftarrow{x}^{'}\text{ if and only if }P\left(\overrightarrow{X}|\overleftarrow{x}\right) = P\left(\overrightarrow{X}|\overleftarrow{x}^{'}\right).
\end{equation}
That is, two pasts $\overleftarrow{x}$ and $\overleftarrow{x}^{'}$ are equivalent if and only if the corresponding conditional probabilities of any future output sequence $\overrightarrow{X}=\overrightarrow{x}$ are identical for all such sequences.
The resulting predictive model, named $\varepsilon$-machine, encodes the given past into a causal state $s_k = \varepsilon(\overleftarrow{x})$~\cite{1986Grassberger,Shalizi2001}.  At each time step, it operates according to a set of transition probabilities $T^x_{kj}:= P(x,s_j|s_k)$--i.e., corresponding to emission of $x$ while transitioning from causal state $s_k$ to $s_j$. $\varepsilon$-machines are a special class of hidden Markov models (HMM) with the synchronizing property--the state of the machine at each time step is fully determined by past observation.

The memory required by an $\varepsilon$-machine is quantified by the Renyi entropies
\begin{equation}
\label{eq:classicl_complexity}
C_\mu^\alpha := \frac{1}{1-\alpha} \log_2\left(\sum_k\pi_k^\alpha\right),
\end{equation}
where $\pi_k$ represents the probability that $\overleftarrow{x}$ belongs to causal state $s_k$. As $\varepsilon$-machines are provably minimal, these Renyi entropies are considered fundamental measures of structure and complexity in the underlying process~\cite{Shalizi2001}. There are two particularly meaningful cases. The most well-studied is the statistical complexity $C_\mu\equiv~C_\mu^1$ ~\cite{Crutchfield1989,perry1999finite} which quantifies the amount of past information a classical predictive model must retain--as measured by Shannon entropy. Meanwhile, the topological complexity $C_\mu^0$~\cite{CrutchfieldTop} represents a single-shot measure of memory cost, capturing the minimal number of configuration states of a predictive model.

\inline{Quantum models}Going beyond (classical) $\epsilon$-machines, quantum models allow for further memory reduction~\cite{gu2012quantum,mahoney2016occam,riechers2015closed,Binder2017}--the most memory-efficient model to date being q-simulator~\cite{Binder2017,mahoney2016occam}. For a given past $\overleftarrow{x}$, a q-simulator retains a quantum state $\ket{\sigma_k}$, encoding the corresponding causal state $s_k$, in its working memory--each with probability $\pi_k$,
 \begin{equation}
 \phi := \sum_k \pi_k\ket{\sigma_k}\bra{\sigma_k}.
 \label{eq:phi}
 \end{equation}
There always exists unitary operator $U$, which interacts this memory register with another quantum system in the initial state $\ket{0}$ such that
 \begin{equation}
 \label{eq:Unitarydefinition}
  U\ket{\sigma_k}\ket{0}
  = \sum_{j,x}\sqrt{T^x_{kj}}\ket{\sigma_j}|x\rangle.
 \end{equation}
Here, $\{\ket{x}\}$ forms an orthonormal basis, in one-to-one correspondence to the elements of the alphabet $\mathscr{A}$. Measurement of the output system in the basis generates the output $x$ with probability $T^x_{kj}$ as desired, collapsing the memory system into the corresponding quantum state $\ket{\sigma_j}$ in the process. Repeated iteration of unitary interaction and measurement generates an output which is statistically identical to the original stochastic process. Because of stationarity the memory register remains in state $\phi$ after each transition~\cite{Binder2017}. The amount of required quantum memory may be quantified by the quantum Renyi entropy:
 \begin{equation}
 	C_q^\alpha := \frac{1}{1-\alpha} \log_2(\mathrm{Tr}[\phi^\alpha]).
 \end{equation}
Analogous to the classical case, $C_q\equiv C_q^1$ represents the average entropic memory required by the q-simulator and we refer to it as the quantum machine complexity. $C_q^0$ quantifies the dimension of the memory Hilbert space $\mathcal{H}_m:= \mathrm{span} \{\ket{\sigma_j}\}$.
Importantly, such quantum encoding almost always leads to a reduction in entropic memory cost ($C_q^1 < C_\mu^1$).

\inline{Matrix Product States}Consider a translationally invariant quantum state of a 1D spin chain with $N$ sites, expressed as
\begin{equation}
 	|\psi\rangle = \sum_{x_1x_2\cdots x_N}c_{x_1x_2\dots x_N}|x_1x_2\cdots x_N\rangle,
\end{equation}
 where $\{|x_k\rangle\}$ is an orthonormal basis of each local Hilbert space. Such a description in terms of coefficients $c_{x_1x_2\dots x_N}$ grows exponentially with $N$. MPSs~\cite{Schollwock2011,June2014} mitigate this adverse scaling by expressing
 \begin{equation}
 	c_{x_1x_2\cdots x_N} = \bra{b_l}A^{x_1} A^{x_2} \cdots A^{x_N}\ket{b_r},\label{eq:coefficients}
 \end{equation}
 where each $A^{x_k}$ is an $m\times m$ complex matrix and $\{\bra{b_l},\ket{b_r}\}$ are the boundaries.  This description gives rise to an intuitive graphical representation (Fig.~\ref{fig:mpsv}).

\begin{figure}[h]
	\includegraphics[scale=0.25]{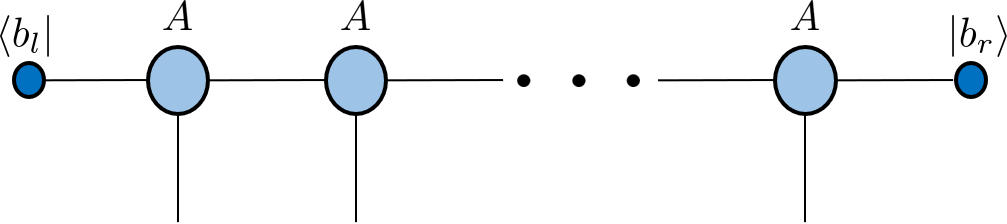}

    \caption{Each site matrix $A$ of a MPS is depicted by a node with ``leg'' representing each of its indices. A link between two nodes indicates summation over the corresponding index.}
    \label{fig:mpsv}
\end{figure}

For $N\to \infty$, we obtain an infinite MPS (iMPS). Many of its properties are determined by the transfer matrix $\mathbb{E}$ (see Fig.~\ref{fig:Domain Eig})
\begin{equation}
	\mathbb{E} = \sum_x A^x\otimes (A^x)^*
\end{equation}
where $(A^x)^*$ is the complex conjugate of $A^x$. If the largest-magnitude eigenvalue $\eta$ of $\mathbb{E}$ is nondegenerate, the boundary becomes irrelevant~\cite{Schollwock2011,June2014}. This is because after applying infinitely many transfer matrices to any potential boundary, it converges to the leading left and right eigenvectors of $\mathbb{E}$, $\mathbb{V}_l$ and $\mathbb{V}_r$--i.e., $\bra{b_l}\bra{b_l}^* \mathbb{E}^\infty \sim \mathbb{V}_l$ and $\mathbb{E}^\infty\ket{b_r}\ket{b_r}^* \sim \mathbb{V}_r$. Note that each leading eigenvector shown in Fig.~\ref{fig:Domain Eig} has two legs. If we designate one index as row index and the other as column index, the leading eigenvectors become matrices, i.e., $V_l$ and $V_r$ (see Supplementary Material~\ref{Appendix:Canonical form}~\cite{SupplementaryInformation}).

\begin{figure}[htp]
  \includegraphics[width=8cm]{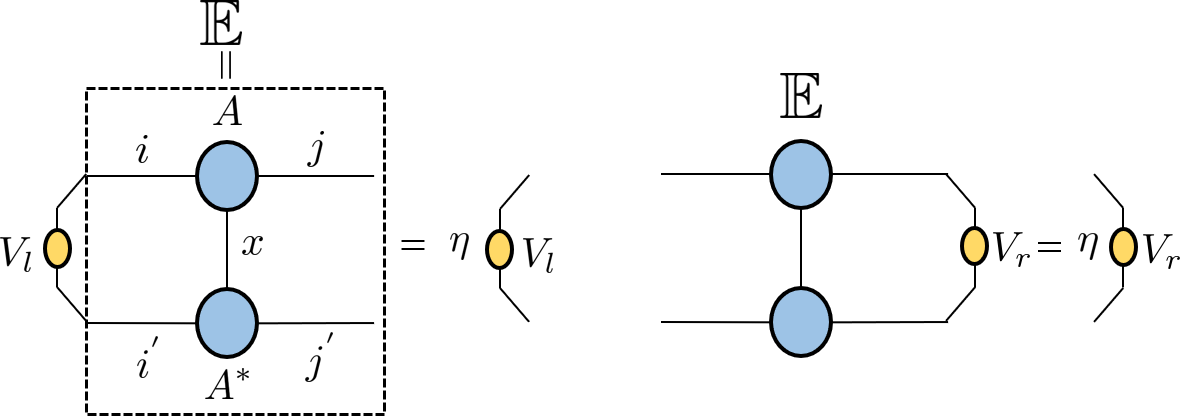}
    \caption{The transfer matrix $\mathbb{E}$ has a largest eigenvalue $\eta$. The corresponding left and right eigenvectors, $V_l$ and $V_r$, are $m\times m$ matrices.}
    \label{fig:Domain Eig}
\end{figure}

\inline{Predictive models and MPS representation}Here we present our main results--the connection between predictive models and the MPS formalism. This is done by associating a stochastic process with a large entangled 1D quantum state
\begin{equation}
  \ket{P\left(\overleftrightarrow{x}\right)} = \sum_{\overleftrightarrow{x}}\sqrt{P\left(\overleftrightarrow{x}\right)}\ket{\overleftrightarrow{x}},
\end{equation}
where $\ket{\overleftrightarrow{x}}$ is an orthonormal basis corresponding to the outputs $\overleftrightarrow{x}:=\cdots x_{-1}x_0\cdots$.  Measuring the quantum state in the basis $\ket{\overleftrightarrow{x}}$ generates the stochastic process. $\ket{P\left(\overleftrightarrow{x}\right)}$ are known as q-samples~\cite{Aharonov2003,Arunachalam2016}, and lie at the heart of many quantum sampling problems~\cite{aaronson2011computational,spring2012boson,crespi2013integrated,lund2014boson,aaronson2016complexity}.

\begin{thm}
A stochastic process $\overleftrightarrow X$, represented by an $\varepsilon$-machine with transition matrix $T$, is fully represented by the iMPS with site matrix
\begin{equation}
	\label{eq:siteMatrix}
	A^x_{kj} = \sqrt{T^x_{kj}}.
\end{equation}
The process is ergodic if and only if the transfer matrix of the iMPS is nondegenerate.
\end{thm}

\begin{proof}
--First, we construct a finite quantum state
\begin{equation}
	\ket{\psi} = \sum_{x_{t:t+L}}\bra{0}A^{x_t}A^{x_{t+1}}\cdots A^{x_{t+L-1}}\ket{0}\ket{x_{t:t+L}}
\end{equation}
with a boundary $\ket{0}$ {and  $x_{t:t+L}:=x_tx_{t+1}\cdots x_{t+L-1}$}. In the limit $L\to \infty, t\to-\infty$, this yields the iMPS. Since each element of each $A^x$ is non-negative the amplitudes generated by the iMPS are as well. In Supplementary Material~\ref{appendix: ergodic property}~\cite{SupplementaryInformation} we prove  that
the largest-magnitude eigenvalue of the transfer matrix of the iMPS given by Eq.~(\ref{eq:siteMatrix}) is nondegenerate if and only if the process is ergodic. Hence, the boundary is irrelevant.

Thus, we only need to show that the amplitudes correspond to the classical distribution $P(\overleftrightarrow{X})$. In Supplementary Material~\ref{Appendix:Canonical form}~\cite{SupplementaryInformation}, the probability distribution generated by the iMPS is shown to be
\begin{equation}
\label{eq:MPS_probability}
\begin{split}
	P_{\mathrm{MPS}}(X_{t:t+L}= x_{t:t+L}) = \\ \hspace{-24pt}\mathrm{Tr}({A^{x_{t+L-1}}}^\dag\cdots&{A^{x_t}}^\dag V_l
   A^{x_t}\cdots A^{x_{t+L-1}}V_r).
\end{split}
\end{equation}

The leading left eigenvector $V_l$ and leading right eigenvector $V_r$ have the following form
\begin{equation}
	\label{eq:Domain_left}
	V_l = \sum_k \pi_k \ket{k}\bra{k}\quad
    V_r = \sum_{k,j} \langle \sigma_k|\sigma_j\rangle \ket{k}\bra{j}
\end{equation}
where $\pi_k$ are the stationary probabilities of causal states $s_k$ and $\ket{j}$ is an orthonormal local site basis (See Supplementary Material~\ref{Appendix: Leftrightleadingeig})~\cite{SupplementaryInformation}). $\ket{\sigma_k}$ are the internal states of the corresponding q-simulator in Eq.~(\ref{eq:Unitarydefinition}). Inserting these expressions into Eq.~(\ref{eq:MPS_probability}), we obtain
\begin{equation}
\label{eq:MPS_distribution}
 P_{\mathrm{MPS}}\left(X_{t:t+L}\right) = P\left(X_{t:t+L}\right),
\end{equation}
as shown in Supplementary Material~\ref{Appendix:probabilities}~\cite{SupplementaryInformation} in more detail. In the limit $L\to\infty$, $t\to-\infty$, this distribution becomes  $P(\overleftrightarrow{X})$.
\end{proof}
This theorem demonstrates that a classical  $\varepsilon$-machine of a stochastic process gives direct rise to an iMPS. This representation has a clear operational meaning: each element of the site matrix corresponds to the transition probability between classical causal states. 

Note that the iMPS in Eq.~(\ref{eq:siteMatrix}) is guaranteed to generate the correct statistics due to $\varepsilon$-machines being predictive models. This follows from Theorem 1, and does not generally hold for nonpredictive HMMs. In addition, ergodicity of the process implies the nondegeneracy of the iMPS. Thus, the q-sample of the process is always the unique ground state of some local Hamiltonian~\cite{2011Norbert}. 

Now, we make use of this MPS representation to point out the relation between the quantum machine complexity and the entanglement of a process' q-sample across any bipartition.

\begin{thm}
 The entanglement across any bipartition of a stochastic process' q-sample equals the quantum machine complexity $C_q$. Its Schmidt rank is equal to the dimension of the memory Hilbert space.
 \label{thm:2}
\end{thm}

\begin{proof}
--First, the leading left and right eigenvectors can always be decomposed as $V_l = W_l^\dag W_l$ and $V_r = W_r W_r^\dag$. The Schmidt coefficients are the positive square roots of the eigenvalues of the density matrix $\rho = W_lW_rW_r^\dag W_l^\dag = W_lV_rW_l^\dag$ (see Supplementary Material~\ref{Appendix:Canonical form}~\cite{SupplementaryInformation}). In line with Eq.~(\ref{eq:Domain_left}) we choose $W_l = \sum_k \sqrt{\pi_k}|k\rangle\langle k|$.   Then we construct a quantum state
\begin{equation}
	\ket{\psi_{AB}} = \sum_k \sqrt{\pi_k}\ket{k}_A\ket{\sigma_k}_B^*
\end{equation}
where $\ket{\sigma_k}^*$ is the complex conjugate of $\ket{\sigma_k}$. Partial trace over system $B$ yields
\begin{equation}
	\rho_A = \sum_{k,j}\sqrt{\pi_k\pi_j} \langle \sigma_k|\sigma_j\rangle \ket{k}\bra{j} = W_l V_r W_l^\dag
\end{equation}
Thus, the Schmidt coefficients of the q-sample are the square roots of the density matrix $\rho_A$'s eigenvalues (see Supplementary Material~\ref{Appendix:Canonical form}~\cite{SupplementaryInformation}).
On the other hand, partial trace over system $A$ gives
\begin{equation}
	\rho_B = \sum_k \pi_k\ket{\sigma_k}^*\bra{\sigma_k}^*.
\end{equation}
Note that the complex conjugate of $\rho_B$ is equal to $\phi$ [Eq.~(\ref{eq:phi})]. Hence, since $\rho_A$ and $\rho_B$ must have the same spectrum, $\rho_A$ also has the same spectrum as $\phi$. Thus, the Schmidt coefficients are the square roots of the eigenvalues of $\phi$.

This implies that any Renyi entropy of the squared Schmidt coefficients and the corresponding entropy of $\phi$ are also equal
\begin{equation}
	H_\alpha\left(c_k^2\right) = H_\alpha(\phi)
\end{equation}
For the special cases $\alpha =1$ and $\alpha = 0$, we obtain Theorem~\ref{thm:2}.
\end{proof}

This connection illustrates an interesting link between the resource costs of generating a time sequence--a property with a temporal direction, with quantum correlations between spatial systems--a property without temporal direction. 

Our final result presents a systematic means of constructing a stochastic process' $q$-simulator from the iMPS of its $q$-sample. 

\begin{thm}
\label{thm:Obtain_Q_simulator}
   The q-simulator for a stochastic process~\cite{Binder2017} can be systematically constructed from its iMPS representation according to Theorem 1:
   \begin{enumerate}
    \item The internal quantum states of the q-simulator are
\begin{equation}
	\ket{\sigma_j} := W_r^\dag |j\rangle
\end{equation}
where $W_r$ is a decomposition  $ W_rW_r^\dag = V_r$ such that the image of $W_r^\dag$ is the same as the image of $V_r$.
\item The stepwise unitary interaction of the q-simulator is given by
\begin{equation}
\label{Eq:Unitary}
  \bra{x}U\ket{0} :=(W_r^{-1}A^x W_r)^\dag,
\end{equation}
 where $W_r^{-1}$ is defined to be the inverse matrix of $W_r$ on the memory Hilbert space $\mathcal{H}_m = \mathrm{span}\{\ket{\sigma_k}\}$.
     \end{enumerate}
\end{thm}

The approach makes use of existing results in MPS literature. There, it was found that each MPS representation of a quantum state gives rise to a means of synthesizing the state through a unitary quantum circuit~\cite{Schon2005,Vidal2007}. In Supplementary Material~\ref{Appendix:Theorem3}~\cite{SupplementaryInformation}, we adapt these methodologies for the specific class of iMPS in Eq.~(\ref{eq:siteMatrix}) such that the resulting quantum circuits are also valid predictive models. That is, the resulting circuit allows a systematic means to encode any given past $\overleftarrow{x}$ into a suitable quantum state $\ket{\sigma_k}$. A prescription for the exact circuit on $\ket{\sigma_k}$ then allows generation of correct conditional future statistics $P(\overrightarrow{X}|\overleftarrow{x})$. In addition, such models feature smaller memory dimension whenever $V_r$ is not full rank.

\inline{Discrete renewal process}We illustrate these results using the discrete renewal process with uniform emission probability~\cite{Marzen2017,Elliott2018Continous}, as defined by its $\varepsilon$-machine in Fig.~\ref{fig:Renewal Process}. 

\begin{figure}[htp]
\includegraphics[scale=0.4]{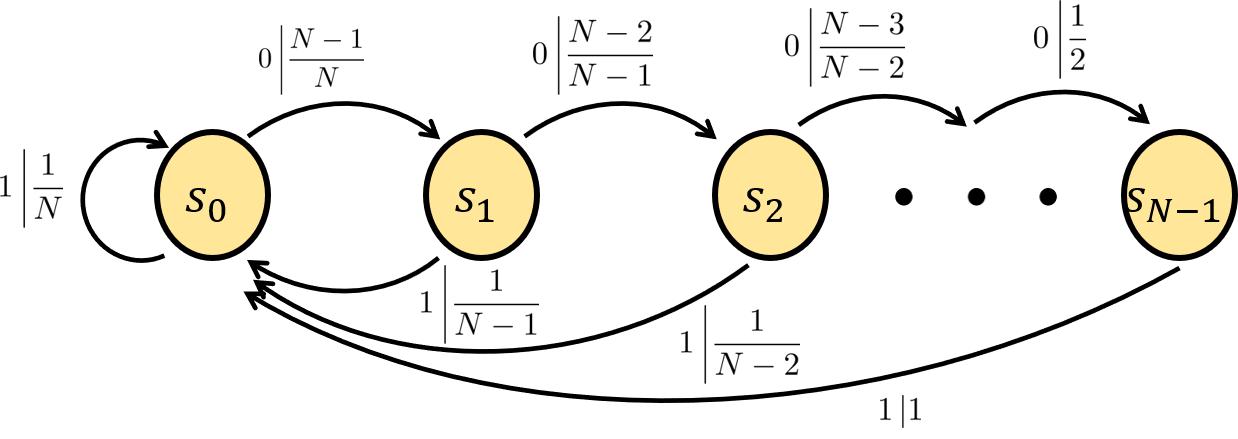}
\caption{{The $\varepsilon$-machine of the discrete renewal process. An edge labeled $x|p$, indicates that the transition from causal state $s_i$ to $s_j$ generates output $x$ with probability $p$. The renewal process either emits $1$ and returns to $s_0$ or $0$ and evolves further.}}
\label{fig:Renewal Process}
\end{figure}

Consider the renewal process' statistics $P(\overleftarrow{X},\overrightarrow{X})$ and its q-sample $\ket{\psi}$. Our results allow (i) direct construction of an MPS representation for $\ket{\psi}$, (ii) the use of MPS methods to find the q-simulator for the renewal process, and (iii) evaluation of the corresponding quantum advantage of simulating such a process quantum mechanically. First, according to Theorem 1, the iMPS site matrix for $\ket{\psi}$ is given by
\begin{equation}
A^0_{k,k+1} = \sqrt{\frac{N-k-1}{N-k}}\quad A^1_{k,0} = \sqrt{\frac{1}{N-k}}
\end{equation}

With Theorem 2, we can compute the quantum memory $C_q$ by using the MPS to quantify the entanglement in $\ket{\psi}$. In Fig.~\ref{fig:Renewal qmemory}, we see that the $q$-simulator requires only a bounded memory while the classical optimal memory scales as $O(\log N)$~\cite{Elliott2018Continous}. In fact,  unbounded memory advantage for $C_q^0$ can also occur~\cite{Thompson2017}.

\begin{figure}[htp]
hi\includegraphics[scale=0.4]{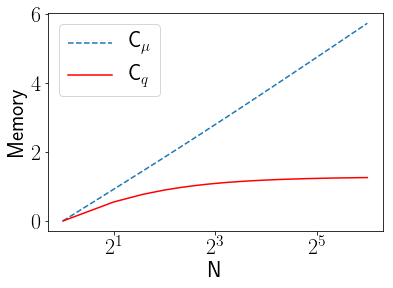}
\caption{The amount of quantum memory $C_q$ converges to a finite value while the optimal classical memory $C_\mu$ diverges logarithmically.}
\label{fig:Renewal qmemory}
\end{figure}

The internal quantum states and the unitary operator $U$ are straightforward to obtain, by virtue of Theorem 3 (see Supplementary Material~\ref{appendix:Renewal process}~\cite{SupplementaryInformation}).

\inline{Conclusion}We have connected two previously distinct notions of complexity: the memory cost of prediction and the representational complexity of spin chains--allowing results in one field to catalyze new insights in the other. By associating each stochastic process with a suitable quantum state of a spin chain, we demonstrated that predictive models constructed in complexity science lead directly to the MPS representation of a general classical stochastic process. Based on this, optimization using standard MPS techniques allows systematic construction of q-simulators--the current state of the art for predictive modeling. Moreover, we established that the memory requirement of such q-simulators exactly coincides with the bipartite entanglement in the associated spin chain state. 

These results open a number of promising avenues for future research. From a foundational perspective, the memory requirements of a q-simulator describe a task that involves a clear temporal direction--using past information to generate future predictions. In the classical setting, this resource cost can vary drastically when time is reversed \cite{Crutchfield2009}, such that costs of prediction and retrodiction differ. Meanwhile, this  research indicates this asymmetry is drastically reduced when quantum models are allowed. The connection to entanglement provides a new perspective to understand this divergence~\cite{Thompson2017}. Moreover, noting that entanglement is calculated rather differently from complexity, the proven equivalence between the two opens the door to identifying more efficient numerical methods for computing the latter.

In addition, tensor networks offer sophisticated techniques to reduce the bond dimension of represented processes at the cost of introducing small inaccuracies. Such methods can now be adapted to quantum modeling--allowing the construction of approximate models with drastically reduced dimensional requirements. Meanwhile, there has been growing interest in generalizing computational mechanics to higher dimensions~\cite{1997Hanson}. It would indeed be interesting to relate such works to higher dimensional variants of tensor networks, such as MERA~\cite{2008Vidal} and PEPS~\cite{Schuch2010}.

\acknowledgements
The authors thank J.\ Thompson, D.\ Poletti, I.\ Arad, J.\ Mahoney, J.\ Crutchfield, and T.\ Elliott for useful discussions. F.C.B. and V.N. thank the Centre for Quantum Technologies for their hospitality. This work was supported by the National Research Foundation of Singapore (NRF-NRFF2016-02), the John Templeton Foundation (Grant No. 54914), the Foundational Questions Institute (grant ``Observer-dependent complexity: the quantum-classical divergence over `what is complex?'{''}),Huawei research, the NRF-ANR grant NRF2017-NRF-ANR004 VanQuTe, and the Singapore Ministry of Education (Tier 1 grant RG190/17).

\bibliography{References}

\clearpage

\appendix
\renewcommand{\thepage}{\roman{page}}
\setcounter{page}{1}
\setcounter{equation}{0}
\setcounter{figure}{0}
\counterwithout{equation}{section}

\renewcommand{\theequation}{S\arabic{equation}}%
\renewcommand{\thefigure}{S\arabic{figure}}%

\onecolumngrid

\section*{Supplementary Information}

\section{MPS and Canonical form}
Consider an iMPS $\{A^x\}$ with left and right boundaries $\{\bra{b_l},\ket{b_r}\}$. In order to link the transfer matrix $\mathbb{E}$ of this iMPS to a completely positive (CP) map we introduce the left and right \textit{vec maps} $\mathcal{V}_l$ and $\mathcal{V}_r$:
\begin{equation}
\mathcal{V}_l:    \langle j|\langle i| \rightarrow \ket{i}\bra{j} \quad
\mathcal{V}_r: \ket{i}\ket{j} \to \ket{i}\bra{j}\label{eq:vecmap}
\end{equation}
A graphical representation of these maps is shown in Fig~\ref{App:Fig:vec}.

\begin{figure}[htp]
	\includegraphics[scale=0.3]{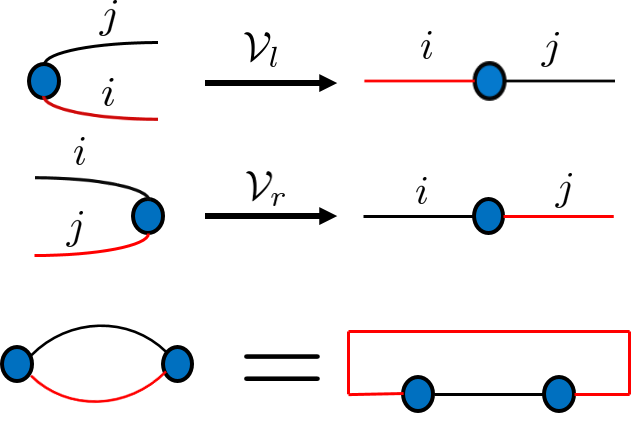}
	\caption{In graphical representation, a vec map is the operation which bends the leg of a node to the other side. }
	\label{App:Fig:vec}
\end{figure}

Applying the left vec map and right vec map to the transfer matrix gives
\begin{equation}
\begin{split}
\leftcp(\ket{j}\bra{i}) &:= \mathcal{V}_l( \langle j|\langle i| \mathbb{E}) = \sum_{x} (A^{x})^\dagger |i\rangle\langle j| A^{x}\\
\rightcp(\ket{j}\bra{i}) &:= \mathcal{V}_l( \mathbb{E}\ket{j}\ket{i}) = \sum_{x} A^{x} |j\rangle\langle i| (A^{x})^\dagger
\end{split}\label{eq:cpmap}
\end{equation}

For an iMPS whose transfer matrix' largest-magnitude eigenvalue is unique, the boundary does not affect the expectation value of any local observable $\mathcal{O}$ on a finite interval of $N$ sites. This can be seen from the fact that any boundary converges to the leading eigenvector of the transfer matrix after applying the transfer matrix infinitely many times (see Fig.~\ref{App:Fig:Expectation value}).
\begin{equation}
\langle\mathcal{O}\rangle = \mathbb{B}_l \mathbb{E}^{\infty}\mathbb{E}_\mathcal{O}\mathbb{E}^\infty \mathbb{B}_r = \mathbb{V}_l\mathbb{E}_\mathcal{O}\mathbb{V}_r
\end{equation}
where $\mathbb{E}_{\mathcal{O}} = \sum_{x_{1:N},x_{1:N}^{'}} \bra{x_{1:N}}\mathcal{O} \ket{x_{1:N}^{'}} (A^{x_1}\cdots A^{x_N})\otimes (A^{x_1^{'}}\cdots A^{x_N^{'}})^*$ and $\mathbb{B}_l = \bra{b_l}\bra{b_l}^*, \mathbb{B}_r = \ket{b_r}\ket{b_r}^*$. Thus, different boundaries give the same expectation value of a local observable. In this case, we can ignore the boundary and focus on the site matrices.

\begin{figure}[htp]
	\includegraphics[scale=0.26]{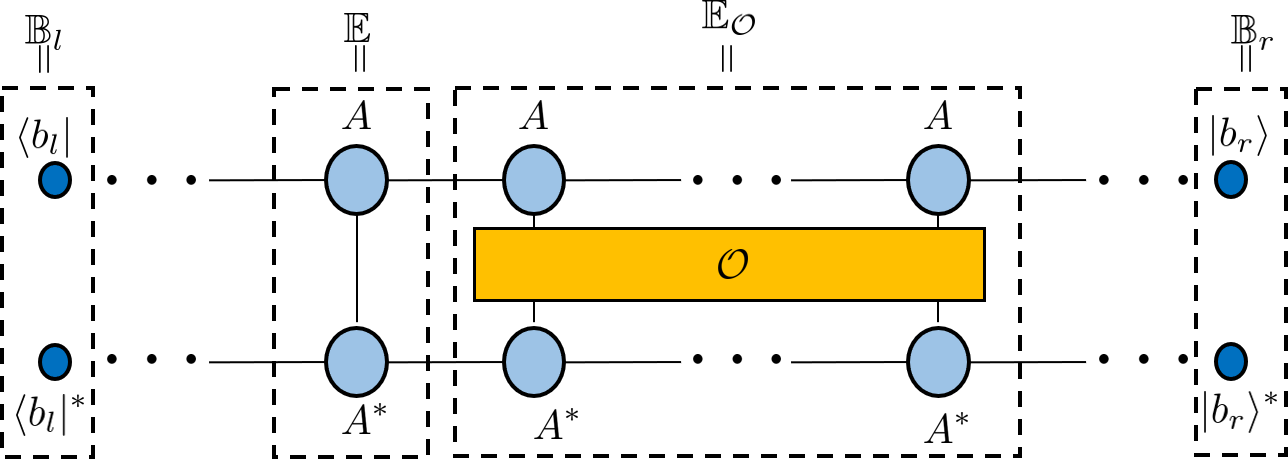}
	\caption{The graphical representation of expectation value $\langle\mathcal{O}\rangle$. $A^*$ denotes the complex conjugate of $A$. The block in the middle represents the observable $\mathbb{E}_\mathcal{O}$, which is a multi-index tensor. The blue nodes are the boundary. The expectation value is given by contracting all indices.}
	\label{App:Fig:Expectation value}
\end{figure}

As the measurement in the computational basis corresponds to the measurement operators $\mathcal{O} =  \ket{x_{t:t+L}}\bra{x_{t:t+L}}$ for finite $t$ and $L$, the distribution of the measurement outcomes is then the expectation value of the measurement operators, i.e.,
\begin{equation}
\begin{split}
P_{\mathrm{MPS}}(X_{t:t+L} = x_{t:t+L}) &= \mathbb{V}_l\mathbb{E}_{\mathcal{O}}\mathbb{V}_r\\
&= \mathrm{Tr} \left[\mathcal{V}_l(\mathbb{V}_r\mathbb{E}_{\mathcal{O}})\mathcal{V}_r(\mathbb{V}_r)\right]\\
&= \mathrm{Tr}\left[{A^{x_{t+L-1}}}^\dag\cdots{A^{x_t}}^\dag V_l A^{x_t}\cdots A^{x_{t+L-1}}V_r\right]
\end{split}
\end{equation}
where $V_l = \mathcal{V}_l(\mathbb{V}_l)$ and $V_r = \mathcal{V}_r(\mathbb{V}_r)$. Note that $\mathbb{V}_l\mathbb{E}_{\mathcal{O}}\mathbb{V}_r = \mathrm{Tr} \left[\mathcal{V}_l(\mathbb{V}_r\mathbb{E}_{\mathcal{O}})\mathcal{V}_r(\mathbb{V}_r)\right] $ is obtained from  the fact that $\bra{i}\bra{j}m\rangle \ket{n} =  \mathrm{Tr}\left[\mathcal{V}_l\left(\bra{i}\bra{j}\right)\mathcal{V}_r\left(\ket{m}\ket{n}\right)\right]$

For any iMPS, a special form, called the canonical form, is expressed as the Schmidt decomposition at any bipartition, see Fig.~\ref{App:fig:canonical form}. For an iMPS whose transfer matrix' largest\hyp magnitude eigenvalue is non\hyp degenerate, there is a systematic way of obtaining the canonical form~\cite{Orus2007}.
\begin{figure}[h]
	\includegraphics[scale=0.4]{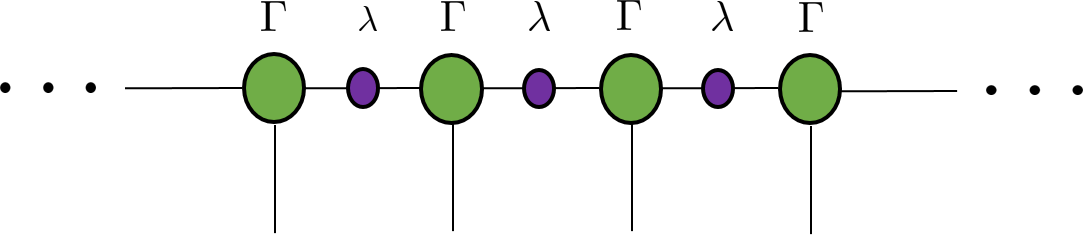}
	\caption{The canonical form of an iMPS consists of two types of tensors, both of which are explained in the text. The purple nodes correspond to the Schmidt coefficients.}
	\label{App:fig:canonical form}
\end{figure}

\label{Appendix:Canonical form}
The essence in computing the canonical form of an iMPS is to deduce its entanglement properties from the transfer matrix $\mathbb{E}$. This is done in four steps:

\begin{enumerate}[label = (\roman*)]
	\item  Assume the largest eigenvalue of transfer matrix $\mathbb{E}$ is unique. Find the leading right and left eigenvectors of the transfer matrix $\mathbb{E}$, $V_r$ and $V_l$, as shown in Fig.~\ref{fig:canonicalization}. Here,  $V_l$ and $V_r$ are rank 2 tensors -- i.e., positive Hermitian matrices. Because the transfer matrix $\mathbb{E}$ is not Hermitian, $V_l$ and $V_r$ are generally different from each other.
	
	\item Decompose $V_r$ and $V_l$ separately, $V_r = W_rW_r^\dagger, V_l = W_l^\dagger W_l$, as shown in Fig.~\ref{fig:canonicalization}. For instance, if $V_r$ has the eigenvalue decomposition $V_r = U D U^\dag$, where $U$ is a unitary matrix and $D$ is a diagonal matrix, then $W_r$ can be set to $W_r = U\sqrt{D}$
	
	\item Insert two identity matrices, $\mathbb{I} = W_rW_r^{-1}$ and $\mathbb{I} = W_l^{-1}W_l$, into the horizontal index, as shown in Fig.~\ref{fig:canonicalization}. Then compute the singular value decomposition of the product $W_lW_r$, namely $W_lW_r = U\lambda V$, where $U$ and $V$ are unitary matrices. $\lambda$ is a diagonal matrix which contains the Schmidt coefficients of $|\psi\rangle$ across any bipartition of the iMPS. \label{App:Schmidt coefficients}
	
	\item Compute $\Gamma^x = V W_r^{-1} A^x W_l^{-1} U$, as shown in Fig.~\ref{fig:canonicalization}.
	
\end{enumerate}

\begin{figure}[h]
	\includegraphics[scale=0.45]{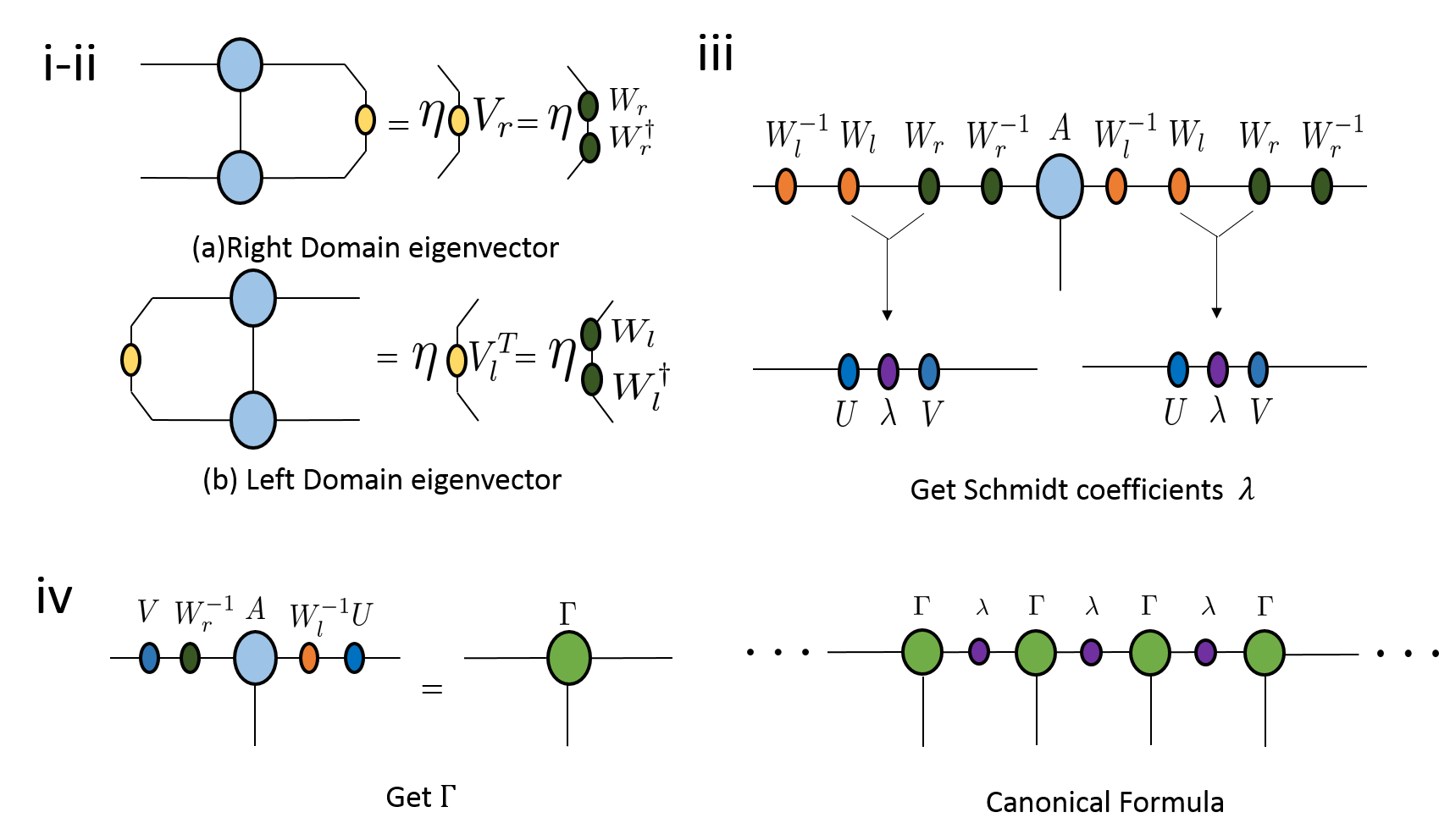}
	\caption{[Color online] The canonical form $\{\Gamma,\lambda\}$ of Fig.~\ref{App:fig:canonical form} is computed in four steps. $\lambda$ is obtained in step (\textrm{iii}) as the singular values of $W_lW_r$, whereas $\Gamma$ is defined in step (\textrm{iv}).}
	\label{fig:canonicalization}
\end{figure}

As shown in~\ref{App:Schmidt coefficients}, the Schmidt coefficients are the singular values of the matrix $W_lW_r$.

\section{The largest eigenvalue of the transfer matrix is unique iff the process is ergodic}
\label{appendix: ergodic property}

In this section, we show that the convergence of the transfer matrix of a q-sample's iMPS representation and convergence of the corresponding stochastic process mutually imply each other (Lemma 2, below). We begin with the following intermediate result:
\begin{lem}
	The largest eigenvalue of the transfer matrix $\mathbb{E}$ is 1.
\end{lem}

First note that the CP map $\leftcp$ given by Eq.~\ref{eq:cpmap} has the same eigenvalues as the transfer matrix $\mathbb{E}$. It hence suffices to show that the largest eigenvalue of $\leftcp$ is 1. For the proof we use the matrix norm $\norm{*}_1$ which is defined by $\norm{B}_1 := \sum_{ij} \abs{B_{ij}}$ for a matrix $B=\sum_{i,j}b_{ij}\ket{i}\!\bra{j}$ in an orthonormal basis $\{\ket{i}\}$.
For any rank-1 matrix $|j\rangle\!\langle k|$,
\begin{equation}
\begin{split}
&\quad\norm{\mathcal{E}_l(|j\rangle\!\langle k|)}_1 \\
&=\sum_{mn}\abs{\langle m|\mathcal{E}_l(|j\rangle\!\langle k|)|n\rangle}\\
&= \sum_{mn}\sum_{x} \sqrt{P(x,s_m|s_j)P(x,s_n|s_k)}
\end{split}
\end{equation}

Since the causal state at the next step is determined by the output symbol and the previous causal state, we obtain

\begin{equation}
\begin{split}
\sum_{mn}\sum_{x} \sqrt{P(x,s_m|s_j)P(x,s_n|s_k)}
&=\sum_{x}\sqrt{P(x|s_j)P(x|s_k)}\\
&\leq 1
\end{split}
\end{equation}

The last step follows from the Cauchy-Schwarz inequality. Equality holds iff $k=j$. For any matrix $B$ as above
\begin{equation}
\begin{split}
\norm{\mathcal{E}_l(B)}_1 \leq \sum_{ij}\abs{b_{ij}}\norm{\leftcp(|i\rangle\!\langle j|)}_1\leq\sum_{ij}\abs{b_{ij}} =\norm{B}_1
\end{split}
\end{equation}  Therefore any eigenvalue $\lambda$ of $\leftcp$ is less or equal to 1.

In general, for any diagonal density matrix $\rho_d=\sum p_k\ket{k}\!\bra{k}$
\begin{equation}
\begin{split}
\leftcp(\rho_d) &= \sum_{x} (A^{x})^{\dagger} \rho_d A^{x}\\
&=\sum_{x}(A^{x})^{\dagger} \sum_k p_k |k\rangle\!\langle k| A^{x}\\
&= \sum_j\sum_k p_k\sum_{x} P(x,s_j|s_k)|j\rangle\!\langle j|
\end{split}
\label{eq:rhod}
\end{equation}
For the stationary distribution in particular we define
\begin{equation}
V_l = \sum_k \pi_k|\ccausal{k}\rangle\!\langle\ccausal{k}|
\end{equation}
which has the following property
\begin{equation}
\leftcp(V_l) = V_l
\end{equation}
This can be seen from Eq.~\ref{eq:rhod} by noting that $\sum_{x,k}P(x,s_j|s_k)\pi_k=\pi_j$ due to stationarity. In other words, $V_l$ is an eigenvector of $\leftcp$ with eigenvalue $1$ -- hence, a leading eigenvector. \hfill$\square$

\begin{lem}
	The transfer matrix has only one eigenvalue, of magnitude one, iff the corresponding process is ergodic.
\end{lem}

First, we assume that the process is ergodic and that $V_l$ is a leading eigenvector of $\leftcp$. Assume that $V_l$ has non-zero off-diagonal elements.
Since $\norm{\leftcp(|j\rangle\!\langle k|)}_1 < 1$ for all $j\neq k$,
\begin{equation}
\norm{\leftcp(V_l)}_1\leq \sum_{jk} \abs{(V_l)_{jk}}\norm{\leftcp(\ket{j}\bra{k})}_1 < \norm{V_l}_1
\end{equation}
which contradicts $\norm{\leftcp(V_l)}_1 = \norm{V_l}_1$. Thus, $V_l$ only has non-zero diagonal elements. Convergence of the process then ensures that there is only one diagonal matrix that is the leading eigenvector of the CP map $\leftcp$.

The inverse case is also true. We consider again Eq.~\ref{eq:rhod}.
Assuming that the transfer matrix's largest-magnitude eigenvalue is unique, $\rho_d$ always converges to the matrix $\sum_k\pi_k\ket{k}\!\bra{k}$. Hence, any distribution under transition converges to the stationary distribution. \hfill$\square$

\section{The structure of leading left and right eigenvectors}
\label{Appendix: Leftrightleadingeig}

We show that the leading left and right eigenvectors have the following expressions
\begin{align}
V_l &= \sum_k \pi_k \ket{k}\bra{k}\label{App: Domain eigenvector V_l}\\
V_r &= \sum_{k,j} \langle \sigma_k|\sigma_j\rangle \ket{k}\bra{j}
\label{App: Domain eigenvector V_r}
\end{align}

As mentioned in Supplementary Information \ref{appendix: ergodic property}, the leading left eigenvector is
\begin{equation}
V_l = \sum_k \pi_k \ket{k}\bra{k}
\end{equation}

For the leading right eigenvector, we show that $V_r$ (Eq.~\ref{App: Domain eigenvector V_r}) is a fixed point of the CP map $\mathcal{E}_r (\rho) = \sum_x A^x\rho(A^x)^\dagger$. Applying $\mathcal{E}_r$ to $V_r$, we obtain
\begin{equation}
\begin{split}
\bra{k}\mathcal{E}_r(V_r)\ket{j} &= \sum_x \bra{k}A^x V_r {A^x}^{\dag}\ket{j}\\
&= \sum_{m,n,x}\sqrt{T^x_{km}T^x_{jn}}\langle \sigma_m|\sigma_n\rangle
\end{split}
\end{equation}
Now we use the fact that $U\ket{\sigma_k}\ket{0} = \sum_x \sqrt{T^x_{km}}\ket{\sigma_m}\ket{x}$, which results in
\begin{equation}
\bra{\sigma_k}\sigma_j\rangle = \sum_{m,n,x}\sqrt{T^x_{km}T^x_{jn}}\bra{\sigma_m}\sigma_n\rangle
\end{equation}

Thus, $\bra{k}\mathcal{E}_r(V_r)\ket{j} = \langle \sigma_k|\sigma_j\rangle$. In other words, $V_r$ is the leading right eigenvector.

\section{iMPS generates the correct distribution}
\label{Appendix:probabilities}
Here, we prove that the iMPS generates the correct distribution. As shown in Supplementary Information~\ref{Appendix:Canonical form}, the distribution of measurement outcomes of an iMPS is
\begin{equation}
\begin{split}
&P_{MPS}(x_{t:L})\\
&= Tr({A^{x_{t+L-1}}}^\dag\cdots{A^{x_{t}}}^\dag V_l A^{x_{t}}\cdots A^{x_{t+L-1}}V_r) \\
&= \sum_n \langle n|{A^{x_{t+L-1}}}^\dag\cdots{A^{x_t}}^\dag V_l A^{x_t}\cdots A^{x_{t+L-1}}\sum_m |m\rangle\langle m| V_r|n\rangle\\
&= \sum_{n,m,i} \pi_j\sqrt{P(x_{t:t+L},s_n|s_i) P(x_{t:t+L},s_m|s_i)} \braket{\sigma_m}{\sigma_n}\\
\end{split}
\end{equation}
where we have used Eq.~\ref{App: Domain eigenvector V_l} and Eq.~\ref{App: Domain eigenvector V_r}. We note that the next causal state is identified by the output symbol $x$ and the previous causal state. In other words, there exists only one causal state $s_n$ for which  $P(x_{t:t+L},s_n|s_i)$ is non-zero.
\begin{equation}
\begin{split}
P_{MPS}(x_{t:L}) &= \sum_{i,n} \pi_jP(x_{t:t+L},s_n|s_i)\braket{\sigma_n}{\sigma_n}\\
&= P(x_{t:t+L})
\end{split}
\end{equation}

\section{Proof of Theorem 3}
\label{Appendix:Theorem3}
There are two aspects of Theorem 3. First, $U$ defined in Eq.~\ref{Eq:Unitary} is a unitary operator. More importantly, $U$ couples the internal quantum states with the ancillary system such that correct statistics are generated, i.e., $U\ket{\sigma_k}\ket{0} = \sum_{x,j} \sqrt{T^x_{kj}}\ket{\sigma_j}\ket{x}$.

For the first part, as $U$ is partly defined on the space $\mathcal{H}_m \otimes \ket{0}$, we show that $U$ is an isometry which can be extended to a unitary operator.
\begin{lem}
	$U$ in Eq.~\ref{eq:Unitarydefinition} is an isometry such that
\end{lem}
\begin{equation}
\begin{split}
\bra{0}U^\dag U\ket{0} &= \sum_x  W_r^{-1}A^xV_r(A^x)^\dag (W_r^{-1})^\dag  \\
&=  W_r^{-1} V_r(W_r^{-1})^\dag \\
&= I
\end{split}
\end{equation}
where $I$ is the identity matrix.
\hfill $\square$

Before showing that $U$ correctly couples the internal quantum states and the ancillary system, we state two lemmata that discuss the image of the operators $A^x$ and $W_r$.
\begin{lem}
	$A^x$ maps a quantum state in memory Hilbert $\mathcal{H}_m$ space to memory Hilbert space, i.e. $A^x\ket{\psi} \in \mathcal{H}_m$ if $ \ket{\psi} \in \mathcal{H}_m$.\label{App:Lem: Ax}
\end{lem}
Assume the leading right eigenvector has the eigenvalue decomposition
\begin{equation}
V_r = \sum_k \lambda_k \ket{\nu_k}\bra{\nu_k}
\end{equation}
where all $\lambda_k$ are positive eigenvalues and $\ket{\nu_k}$ is the eigenvector. Assume the lemma is wrong. This implies that there exists a quantum state $\ket{\psi}$ that is not in the memory Hilbert space but satisfies
\begin{equation}
\bra{\psi}A^x\ket{\nu_s} \neq 0
\end{equation}
for some $\ket{\nu_s}$ since the basis $\ket{\nu_k}$ spans the memory Hilbert space. Hence
\begin{equation}
\begin{split}
\bra{\psi}V_r\ket{\psi} &= \bra{\psi}\sum_x A^x V_r(A^x)^\dag\ket{\psi}\\
&= \sum_{x,s} \lambda_s|\bra{\psi}A^x\ket{\nu_s}|^2\\
& \geq \lambda_s |\bra{\psi}A^x|\nu_s\rangle|^2\\
& > 0
\end{split}
\end{equation}
However, $\bra{\psi}V_r\ket{\psi} = 0$ contradicts the above equation. Thus, the lemma is correct. \hfill $\square$

Then, we prove that
\begin{lem}
	$W_r$ maps any quantum state into the memory Hilbert space.\label{App:Lem: wr}
\end{lem}
$W_r$ has the singular value decomposition
\begin{equation}
W_r = \sum_i c_i\ket{u_i}\bra{v_i}
\end{equation}
where $c_i$ is the singular value. $\{\ket{u_i}\}$ and $\{\bra{v_i}\}$ are orthogonal bases. $\{\ket{u_i}\}$ spans   the image space of $W_r$. Thus $V_r = \sum_i c_i^2 \ket{u_i}\bra{u_i}$ has the same image space as $W_r$. And the memory Hilbert space, which is the image of map $V_r$, is equal to the image of $W_r$. \hfill $\square$

Combining Lemma~\ref{App:Lem: Ax} and Lemma~\ref{App:Lem: wr}, we have
\begin{equation}
\mathcal{P}_{\mathcal{H}_m}A^x W_r = A^x W_r
\end{equation}
where ${P}_{\mathcal{H}_m}$ is a projector onto the memory Hilbert space $\mathcal{H}_m$.

Applying $U$ on the quantum state and the ancillary system, we have
\begin{equation}
\begin{split}
U\ket{\sigma_k}\ket{0} &= \sum_{x} (W_r^{-1}A^xW_r)^\dag \ket{\sigma_k}\ket{x}\\
&= \sum_x (P_{\mathcal{H}_m}  A^xW_r)^{\dag}\ket{k}\ket{x}\\
&= \sum_{x,j} \sqrt{T^x_{kj}}\ket{\sigma_j}\ket{x}
\end{split}
\end{equation}
\hfill $\square$

\section{Renewal process}
\label{appendix:Renewal process}
In this section we explicitly derive the internal quantum states $\ket{\sigma_j}$ and the unitary operator $U$ for the renewal process. First, note that $W_r^\dag$ in the basis $\ket{j}$ is given by the matrix
\begin{equation}
\begin{bmatrix}
\sqrt{\frac 1N} & 0 & \cdots & 0 \\
\sqrt{\frac 1N} & \sqrt{\frac {1}{N-1}} &\cdots &0 \\
\vdots & \vdots & \ddots & \vdots\\
\sqrt{\frac 1N} & \sqrt{\frac {1}{N-1}} & \cdots & 1
\end{bmatrix}
\end{equation}
Thus, the signal state is
\begin{equation}
|\sigma_k\rangle = W_r^\dagger|k\rangle =\sum_{j\geq k} \sqrt{\frac{1}{N-k}}|j\rangle
\end{equation}
This leads to the canonical form of the iMPS and the desired elments of the unitary operator $U$:
\[
\bra{0}U\ket{0} :=(W_r^{-1}A^0 W_r)^\dag = \sum_{j=0}^{N-2} \ket{j+1}\bra{j} =\begin{bmatrix}
0 & 0 & 0 &\cdots& 0\\
1 & 0 & 0 &\cdots& 0\\
0 & 1 & 0 &\cdots& 0\\
\vdots & \vdots &\vdots & \ddots &\vdots\\
0 & 0 & 0 & \cdots & 0
\end{bmatrix}\quad \bra{1}U\ket{0} :=(W_r^{-1}A^1 W_r)^\dag =\ket{\sigma_0}\bra{N} = \begin{bmatrix}
0 & 0 & 0 &\cdots& \sqrt{\frac 1N}\\
0 & 0 & 0 &\cdots& \sqrt{\frac 1N}\\
0 & 0 & 0 & \cdots &\sqrt{\frac 1N}\\
\vdots & \vdots &\vdots & \ddots &\vdots\\
0 & 0 & 0 & \cdots & \sqrt{\frac 1N}
\end{bmatrix}
\]
\end{document}